\documentclass[12pt,reqno]{amsart}

\usepackage{fullpage, graphicx, amsmath, amssymb}
\usepackage{blindtext}

\theoremstyle{plain}
 \newtheorem{thm}{Theorem}[section]
 \newtheorem{prop}[thm]{Proposition}
 \newtheorem{lem}[thm]{Lemma}
 \newtheorem{cor}[thm]{Corollary}

\theoremstyle{definition}

 \newtheorem{exm}{Example}[section]
 \newtheorem{dfn}{Definition}[section]

 \numberwithin{equation}{section}

\title{Diffeomorphisms of Scalar Quantum Fields via Generating Functions}
\author{Ali Assem Mahmoud\;\;}
\address{Ali Assem Mahmoud: Department of Mathematics, Faculty of Science, Cairo University, Egypt; and Department of Combinatorics and Optimization, University of Waterloo, ON, Canada}
\email{ali.mahmoud@uwaterloo.ca}
\author{\;\;Karen Yeats}
\address{Karen Yeats: Department of Combinatorics and Optimization, University of Waterloo, ON, Canada}
\email{kayeats@uwaterloo.ca}
\thanks{KY is supported by an NSERC Discovery grant, and was supported by the Humboldt Foundation as a Humboldt fellow during the development of this work.  KY would like to thank Dirk Kreimer and Humboldt University for hosting her visit to Berlin as a Humboldt Fellow.  Both authors would like to thank Dirk Kreimer and Paul Balduf for many relevant discussions.  Thanks also to the referee for their comments and insights.}

\begin{document}

\maketitle
\begin{abstract}
We study the application of formal diffeomorphisms to scalar fields. We give a new proof that interacting tree amplitudes vanish in the resulting theories.  Our proof is directly at the diagrammatic level, not appealing to the path integral, and proceeds via a generating function analysis so is more insightful than previous proofs.  Along the way we give new combinatorial proofs of some Bell polynomial identities, and we comment on the connection with the combinatorial Legendre transform.
\end{abstract}

\section{Introduction}
A free scalar quantum field is usually defined via a Lagrangian density 
\[
L(\phi) = \frac{1}{2}\partial_\mu \phi(x)\partial^\mu \phi(x) - \frac{m^2}{2}\phi^2(x)
\]
that contains no self-interaction terms, where $m$ is the mass of the $\phi$-particle. A \textit{field diffeomorphism} $F$ is going to be formally defined as a power series in the field
\[\phi\mapsto F(\phi)=a_0\phi(x)+a_1\phi(x)^2+\cdots \;= \sum_{j=0}^\infty a_j \phi^{j+1},\]
where $a_0=1$, i.e. $F$ is a formal diffeomorphism tangent to the identity. The problem is then to study the field theory expressed by the transformed Lagrangian, if one applies the diffeomorphism to the Lagrangian equation above. The result is seemingly an interacting theory. 

Similarly, one can take an interacting scalar field theory and apply a field diffeomorphism, again resulting in many new interaction terms in the Lagrangian.  In both cases one would expect the contributions of these new terms should cancel.  Likewise, as we will be treating $F(\phi)$ as a formal power series, $\phi$ could in fact be a vector field rather than a scalar field with the $a_n\phi^n$ interpreted as symmetric tensors.  All we really need is that $F(\phi)$ behaves as a formal power series in the indeterminate $\phi$.

In classical field theory a field diffeomorphism is merely a canonical transformation that does not change the Poisson brackets \cite{paulrefers}, and it simply relates theories with different Lagrangians. However, for quantum fields, there are some ambiguities, probably due to operator ordering in the path-integral formulation, and the topic is therefore controversial \cite{diff1,diff2,diff3,diff4,diff5,diff6,diff7}.  Additionally, the order by order cancellations take a quite complicated form; one cannot find small sets of diagrams which cancel, but can only see it in the final sum.

In fact the expected cancellations do hold, a result which we give a new, direct proof of, bypassing issues with the path integral and cleaning up a previous intricate and uninsightful proof of one of us with Dirk Kreimer.

\section{Motivation and Prior Work}

The approach followed in \cite{kreimer2} and \cite{karendiffeo} is a `least-action' approach:  they study field diffeomorphisms order by order in perturbation theory. In \cite{kreimer2} D. Kreimer and A. Velenich showed by direct calculations that, up to six external legs, interacting tree-level amplitudes do vanish. Yet, it was not still known how this can be generalized to higher orders. The vanishing of tree-level amplitudes is crucial as it leads to the vanishing of loop amplitudes via Cutkosky rules and the optical theorem. In \cite{karendiffeo} one of us with Dirk Kreimer proved that if a point field diffeomorphism $\phi(x) \mapsto \sum_{j\geq 0} a_j \phi^{j+1}$ is applied to a free scalar field theory, the resulting field theory, while it appears to have many interaction terms, in fact remains a free theory by appropriate cancellations between diagrams. 

At tree level these cancellations hold whenever the external edges are on-shell, while at loop level they additionally require renormalization with a kinematical renormalization scheme.  This work followed up on the observations by Kreimer and Velenich \cite{kreimer2}.

The arguments of \cite{karendiffeo} proceeded first to reduce the tree level problem to a purely combinatorial problem of proving certain combinatorial identities. These were then proved using Bell polynomials.  Then the loop level results were bootstrapped off the tree level results using Cutkosky rules and the optical theorem.  For both the tree level and loop level results, the key thing to consider was the tree level amplitudes with exactly one external edge which was potentially off-shell.  Calculating these one off-shell edge amplitudes is what reduces to a purely combinatorial problem, and what, with the optical theorem, glues up into the loop level results.

However, the proofs of \cite{karendiffeo}, even at tree level, were unsatisfying as they were both opaque and intricate, consisting of delicate Bell polynomial manipulations which needed to reach fairly deeply into the repertoire of known Bell polynomial identities without obtaining insight.  The authors in \cite{karendiffeo} conjectured that a proof on the level of generating functions could be possible, and could give better insight, especially given the fact that Bell polynomials come from series composition. 

\textbf{Our Contribution:}
This is what we do in this paper, reproving the tree level cancellations of \cite{karendiffeo} at the level of generating functions, and then leveraging the extra insight gained to see exactly how the solution appears as a compositional inverse, and making an explicit connection with the combinatorial Legendre transform of Jackson, Kempf, and Morales.  The latter is particularly interesting because of the role of the on-shell condition in the outcome of the combinatorial Legendre transform in our situation.  
The tree-level amplitudes with at most one off-shell edge remain key for us, since we use the same reduction to combinatorics as \cite{karendiffeo}.  
We show that the series which is the exponential generating function of these tree-level amplitudes with at most one off-shell edge for the transformed theory is exactly the compositional inverse of the diffeomorphism $F$ that was originally applied. Additionally, along the way we give new combinatorial proofs of some Bell polynomial identities due to Cvijovi\'c (see \cite{cvijovic8}).

Note that ultimately every problem considered in this paper is purely combinatorial.  

In terms of more physical considerations, note that no appeal to the path integral or its measure is used in \cite{karendiffeo} nor here.  All results are proven by rigorous arguments at the diagram level.  Consequently these results are ground truth, and the correct transformations for the path integral and path integral measure can be reverse engineered from them.  That a field diffeomorphism ought to pass nicely though the path integral is often viewed as a near triviality, though others have argued that in fact it does not (see \cite{diffeonot}). Different lines of thought can also be seen in \cite{diff1,diff2,diff3,diff4,diff5,diff6,diff7}.  Settling this rigorously while side stepping the path integral entirely was a major motivation for \cite{karendiffeo} as well as for us here.  Furthermore, even from a physical perspective where this result is clear, our approach makes explicit exactly how the non-trivial cancellations work in order to give the diffeomorphism invariance.

Paul Balduf independently arrived at the fact that the series from the tree-level amplitudes is the compositional inverse of the diffeomorphism through analyzing the $S$-matrix \cite{Paul, Paulthesis}. Both he and us first obtained this fact in the fall of 2018, and we discussed our different proofs at that time, while the present authors were visiting Berlin.  One outcome of these discussions is that our method was used in the proof of Theorem 3.3 in \cite{Paul}.

\section{Field Theory Set-up}\label{sec field theory set up}
Let $F$ be a field diffeomorphism 
\[
F:\phi\mapsto F(\phi)=\sum_{j=0}^\infty a_j\phi^{j+1}
\]
with $a_0=1$. When $F$ is applied to a free field $\phi(x)$ with Lagrangian density 
\[
L(\phi)=\frac{1}{2}\partial_\mu\phi(x)\partial^\mu\phi(x)-\frac{m^2}{2}\phi^2(x),
\]
it gives the new Lagrangian 
\[
L_F(\phi)=\frac{1}{2}\partial_\mu F(\phi)\partial^\mu F(\phi)-\frac{m^2}{2}F(\phi)F(\phi),
\] 
where the field $\phi$ is a scalar field from the 4-dimensional Minkowski space-time ($\phi:\mathbb{R}^4\rightarrow\mathbb{R}$).
Expanding out the transformed Lagrangian we obtain
\begin{align*}
L_F(\phi)&=\frac{1}{2}\partial_\mu F(\phi)\partial^\mu F(\phi)-\frac{m^2}{2}F(\phi)F(\phi) \\
    & = \frac{1}{2}\partial_\mu (\phi+a_1\phi^2+a_2\phi^3+\cdots)\partial^\mu (\phi+a_1\phi^2+a_2\phi^3+\cdots)-\frac{m^2}{2}(\phi+a_1\phi^2+a_2\phi^3+\cdots)^2 \\
    & = \frac{1}{2}\partial_\mu \phi\partial^\mu \phi+ a_1\frac{1}{2}\partial_\mu \phi^2\partial^\mu \phi+a_1\frac{1}{2}\partial_\mu \phi\partial^\mu\phi^2+a_1^2\frac{1}{2}\partial_\mu \phi^2\partial^\mu \phi^2 \\
    & \quad + a_2\frac{1}{2}\partial_\mu \phi^3\partial^\mu \phi+a_2\frac{1}{2}\partial_\mu \phi\partial^\mu\phi^3+\cdots -\frac{m^2}{2}\phi^2 - (2a_1)\frac{m^2}{2}\phi^3 - (a_1^2+2a_2)\frac{m^2}{2}\phi^4 - \cdots \\
    & = \frac{1}{2}\partial_\mu \phi\partial^\mu \phi+ (4a_1)\frac{1}{2}\phi \partial_\mu \phi\partial^\mu \phi+(4a_1^2+6a_2)\frac{1}{2}\phi^2\partial_\mu \phi\partial^\mu\phi + \cdots  \\
    & \quad -\frac{m^2}{2}\phi^2 - (2a_1)\frac{m^2}{2}\phi^3 - (a_1^2+2a_2)\frac{m^2}{2}\phi^4 - \cdots \\
    & = \frac{1}{2}\partial_\mu \phi\partial^\mu \phi+ \frac{1}{2}\partial_\mu \phi\partial^\mu \phi\sum_{n=1}^\infty\frac{d_n}{n!}\phi^n - \frac{m^2}{2}\phi^2 - \frac{m^2}{2}\sum_{n=1}^{\infty}\frac{c_n}{(n+2)!}\phi^{n+2}. \end{align*}
where $d_n = n!\sum_{j=0}^n (j+1)(n-j+1)a_ja_{n-j}$ and $c_n = (n+2)!\sum_{j=0}^n a_ja_{n-j}$, (see equation 15 of \cite{kreimer2} for a formulation with slightly different conventions). 

We can see that from each term of the original free Lagrangian we obtain a vertex of each order $\geq 3$, (thus we have two types of vertices of each order) which we will call the kinematic and massive vertices respectively.  We read off the Feynman rules to be 

\begin{itemize}
    \item 

\[
    i\frac{d_{n-2}}{2}(p_1^2+p_2^2+\cdots + p_n^2)
\]
for the $n$-point kinematic vertex where $p_1, \ldots, p_n$ are the momenta of the incident edges; and 

\item 
\[
    -i\frac{m^2}{2}c_{n-2}
\]
for the $n$-point massive vertex.  

\item The free part of the Lagrangian is unchanged so the propagator remains
\[
\frac{i}{p^2-m^2}
\]
for momentum $p$.  We are interested in the on-shell $n$-point tree level amplitude.
\end{itemize}

For the combinatorial reader let us spell out in a bit more detail how the above leads to a purely combinatorial problem on trees.  We are working with graphs with external edges.  For a graph theorist such graphs can be constructed as bipartite graphs where if the bipartition is $(A,B)$ then we require that all vertices in $B$ are either of degree $1$ or $2$ and all vertices in $A$ are of degree $\geq 3$.  Then $B$ in fact contains no additional information: the 2-valent vertices of $B$ just mark the \emph{internal edges} of the original graph, while the 1-valent vertices of $B$ mark some bare half-edges of the original graph, known as \emph{external edges} or \emph{legs}.

To calculate the $n$-point tree level amplitude, we must sum over all trees (connected acyclic graphs of the type above) with $n$ external edges and with vertices either kinematic or massive.  For each tree we compute as follows. To each internal and external edge of the tree assign a momentum $p$ in Minkowski space, that is in $\mathbb{R}^4$ but using the pseudo-metric $|(a_0, a_1, a_2, a_3)|^2 = -a_1^2 + a_2^2+a_3^2+a_4^2$ (in fact the choice of signature will not matter).  Following the usual convention we will write $p^2$ for $|p|^2$.  Impose momentum conservation at each vertex, that is the sum of the momenta for the edges incident to any given vertex must be $0$.  Now multiply the factors given by the Feynman rules for each vertex and the propagator for each internal edge to get the contribution of this tree.

The on-shell condition applies only to the external edges and this condition is that $p^2=m^2$ for each external momentum $p$.  Because we are working with a pseudo-metric, note that $p^2=0$ does not imply $p=0$.

In summary, combinatorially we have the following:
\begin{enumerate}
    \item The $n$-point tree level amplitude is the sum over all trees with $n$ external edges. 
    \item Each external edge is labelled, nothing else is.
    \item A momentum variable $p$ is assigned to every internal and external edge.
    \item The on-shell condition is that $p^2=m^2$ holds for the momentum of every external edge, where ($p^2=p\cdot p$).
    \item Conservation of momenta holds at every vertex.
    \item
    The vertices come in two kinds, massive and kinematic, each with its own contribution to the sum given by the Feynman rules. 
    \item The Feynman rule for a kinematic vertex of degree $n$ with momenta $p_1,\ldots,p_n$ for the incident legs is
    \[i\frac{d_{n-2}}{2}(p_1^2+p_2^2+\cdots+p_n^2),\]where
    \[d_r=r!\sum_{j=0}^r(j+1)(r-j+1)a_ja_{r-j}. \]
    \item The Feynman rule for a massive vertex of degree $n$ is
    \[-i\frac{m^2}{2}n!\sum_{j=0}^{n-2}a_ja_{n-2-j}.\]
    
    \item  The Feynman rule for an internal edge is \[\frac{i}{p^2-m^2},\] where $p$ is the momentum assigned to this propagator.
    
    \item  The contribution of each tree is the product of the Feynman rules for its vertices and internal edges (no contribution from external edges).
\end{enumerate}

Now notice that since we are summing over all such trees, we get the same value if we consider only a single type of combined vertices each of which is the sum of kinematic and massive vertices of degree $n$, for each $n$.

\subsection{Tree-level Amplitudes}
The best way to explore the problem combinatorially is through a small example. Fix an internal edge $e$, and consider all the possible subtrees that may occur below $e$ for a fixed number of legs, $n$. Let $b_n$ be the sum over all such subtrees with the Feynman rules applied to the vertices and edges of the subtree along with the edge $e$ itself.

\begin{exm}
When $n=3$ we have the contributions from the tree graphs in Figure \ref{diffeotr} below.

\begin{figure}[h!]
    \centering
    \includegraphics[scale=0.75]{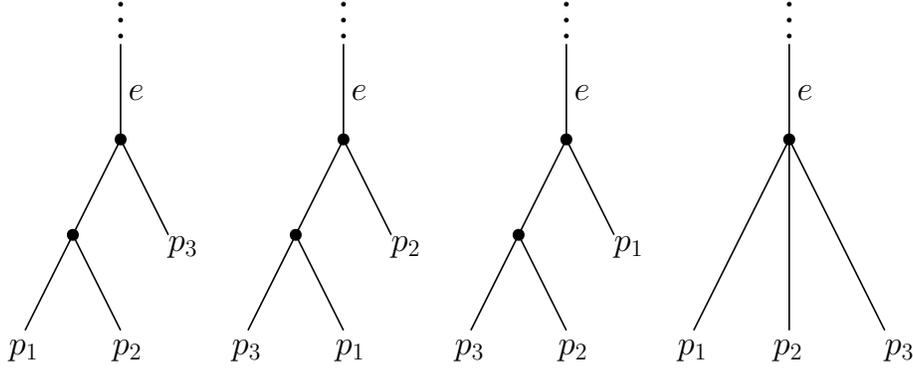}
    \caption{Subtrees below $e$ corresponding to $n=3$ external legs.}
    \label{diffeotr}
\end{figure}

 where the vertical dots above $e$ indicate the rest of the tree.
 Let us recursively compute $b_3$, knowing that $b_2=-2a_1$ (and $b_1=1$) which is easily verified. We shall denote the sum $p_1+p_2+p_3$ by $\mathbf{p}$.
 
 $$b_3=\frac{i^2(\frac{d_2}{2}(p_1^2+p_2^2+p_3^2+\mathbf{p}^2)+c_2)}{\mathbf{p}^2-m^2}+\frac{i^2(\frac{d_1}{2}(p_1^2+(p_2+p_3)^2+\mathbf{p}^2)+c_1)b_2}{\mathbf{p}^2-m^2}+$$
 
 $$+\frac{i^2(\frac{d_1}{2}(p_2^2+(p_1+p_3)^2+\mathbf{p}^2)+c_1)b_2}{\mathbf{p}^2-m^2}+\frac{i^2(\frac{d_1}{2}(p_3^2+(p_1+p_2)^2+\mathbf{p}^2)+c_1)b_2}{\mathbf{p}^2-m^2}.$$
 
 This simplifies to $b_3=-6a_2+12a_1^2$ as the reader may check. It is these cancellations that required an interpretation and triggered the research in \cite{karendiffeo,kreimer2}, especially because of the independence from momenta and masses in the resulting values.

\end{exm}

Returning to the general case, we are interested in the quantity $b_n$, which is the result of fixing an internal edge $e$ and summing over all possible subtrees with $n$ external edges labelled $p_1, p_2, \ldots,p_n$.  Equivalently $b_n$ is the sum over all trees with $n+1$ external edges where one of these edges (called $e$) is not necessarily on-shell, and where additionally we include the propagator for $e$ itself in each term.  

Fortunately, $b_n$ can be computed recursively. Consider the edges below $e$ which are incident to $e$.  Each of them is either external or has another subtree rooted at their other end.  Summing over all possibilities below $e$ means summing over all possibilities for each of these subtrees, and hence the contribution of the subtrees is itself a smaller $b_i$.  By induction, this gives the following recurrence, see \cite{karendiffeo} for a full proof:
\begin{equation}\label{eq orig b rec}
\begin{gathered}
b_n=-\underset{\underset{P_i\neq\emptyset\;\text{and disjoint}}{P_1\cup\cdots\cup P_k=\{1,\ldots,n\}}}{\sum}b_{|P_1|}\cdots
b_{|P_k|}\times \\\frac{\frac{(k-1)!}{2}\underset{j=0}{\overset{k-1}{\sum}}a_ja_{k-1-j}\Big(-m^2(k+1)k+(j+1)(k-j)\big(\sum^k_{i=1}(\sum_{e\in P_i}p_e)^2+(\sum^n_{s=1}p_s)^2\Big)}{(\sum^n_{s=1}p_s)^2-m^2}.\vspace{0.6cm}\end{gathered}\end{equation}
 
 The idea then was to break this into two recurrences and it turned out that some intricate Bell polynomial identities give one way to solve for $b_n$. The approach in \cite{karendiffeo} makes an extensive use of Bell polynomials identities on different levels. In our case we will show that by sticking to working with the exponential generating series of the $b_i$ the use of Bell polynomial identities can be substantially minimized way. Besides, we will give new proofs for a number of these identities. For example, in the Section \ref{bellpolyidsection} we give a simple combinatorial proof for the recent identity obtained by Cvijovi\'c (see \cite{cvijovic8}) in 2013.

 \section{The Role of Bell Polynomials}
 
 \begin{dfn}[Partial Bell Polynomial]\label{bellpoly} The \textit{partial Bell polynomial}, for parameters $n,k$, in an infinite set of indeterminates, $x_1, x_2, \ldots$, is defined by
\begin{align*}B_{n,k}(x_1,x_2,\ldots) &= \underset{P_i\text{'s disjoint, nonempty   }}{\underset{P_1\cup\cdots\cup P_k=\{1,\ldots,n\}}{\underset{\{P_1,\ldots,P_k\}}{\sum}}}x_{|P_1|}x_{|P_2|}\cdots x_{|P_k|}\\&=\underset{\lambda(n,k)}{\sum}\frac{n!}{j_1!j_2!\cdots} \left(\frac{x_1}{1!}\right)^{j_1}\left(\frac{x_2}{2!}\right)^{j_2}\cdots, \end{align*}

where the second sum ranges over all partitions $\lambda=1^{j_1}2^{j_2}\ldots$ of $n$ with $k$ parts, that is, such that \[j_1+j_2+j_3+\cdots=k\;\;\;\text{and}\,\,\,\;j_1+2j_2+3j_3+\cdots=n\;\;\;\text{and}\;\;\; j_i\geq1.\]
\end{dfn}

  Note that, by this definition, the largest index appearing should be $x_{n-k+1}$, thus, any given Bell polynomial uses only a finite number of variables (and is indeed a polynomial). 
  
  On the level of generating functions, one gets
\[ \exp\bigg(u \underset{m\geq1}{\sum}x_m\frac{t^m}{m!}\bigg)=\underset{n,k\geq0}{\sum}B_{n,k}(x_1,x_2,\ldots)\;\frac{t^n}{n!}u^k .\]  This can be used as an alternative definition for Bell polynomials.
 
 To split equation \eqref{eq orig b rec} up usefully, we will use the fact that the problem is symmetric in the external momenta along with the on-shell condition.  Consider expanding all the $p_e^2$ in the numerator and the denominator into sums of squares of external momenta and dot products of distinct external momenta.  All the squares of external momenta are $m^2$ by the on-shell condition, so in both the numerator and denominator collect all of these along with the explicit $m^2$.  The remaining terms all have a factor which is a dot product of distinct external momenta.  By the symmetry in the external momenta, we know that each dot product appears the same number of times, so it suffices to keep track of how many dot product terms there are, without keeping track of which momenta are involved.  So to satisfy equation \eqref{eq orig b rec} it suffices to separately satisfy the $m^2$ part and the dot product part of it.  These two parts are respectively the equations of the following Lemma (see \cite{karendiffeo} for details). 
 
 \begin{lem}{\cite{karendiffeo}} Let $b_n$ be, as before, the sum over all amplitudes of rooted trees with $n+1$ external legs, one of which is off-shell and has a propagator contribution (see the Feynman rules above). Then the sequence $\mathbf{b}=(b_n)$ satisfies equation \eqref{eq orig b rec} if and only if it satisfies the following two recurrences :
 
 \begin{align}
     0&=\overset{n}{\underset{k=1}{\sum}}B_{n,k}(\mathbf{b})\;\frac{(k-1)!}{2}\;\;\overset{k-1}{\underset{j=0}{\sum}}a_ja_{k-1-j}\;\Big[\;2n(j+1)(k-j)-k(k+1)\;\Big]\label{recurrence}\;,\\
 0&=\overset{n}{\underset{k=1}{\sum}}\overset{k-1}{\underset{j=0}{\sum}}a_ja_{k-1-j}(j+1)(k-j)\frac{(k-1)!}{2k}\overset{n}{\underset{s=1}{\sum}}\frac{b_s}{s!(n-s)!}B_{n-s,k-1}(\mathbf{b})(ks(s-1)+n(n-1));\label{di}
 \end{align}
  where $\mathbf{b}=(b_1,b_2,\cdots)$.\end{lem}

  \setlength{\parindent}{0.5cm}
  From these two equations, Karen Yeats and Dirk Kreimer proved that  \[\tag{$\dagger$}b_{n+1}=\overset{n}{\underset{k=1}{\sum}}\frac{(n+k)!}{n!}B_{n,k}(-1!a_1,-2!a_2,\ldots)\;\;,\label{form} \] which might be suggested by the calculation of the first examples of $b_n$'s. 
  
  Proving this formula for the $b_{n+1}$ is all that is needed to prove that the on-shell tree-level amplitudes of the transformed theory are $0$ for $3$ or more external edges.  This is because the $b_{n+1}$ it already almost the tree-level amplitude for $n+2$ external edges -- the only differences are that the propagator for edge $e$ was included and edge $e$ was not on-shell  in $b_{n+1}$.  In particular, then, the on-shell tree-level amplitude is $(p_e^2 - m^2)b_{n+1}$ with the $b_{n+1}$ independent of masses and momenta; hence the amplitude is $0$ when edge $e$ is on-shell.
  
  The $b_{n}$ are also used in the proof of the cancellation of the loop amplitudes, see \cite{karendiffeo} for the argument.

  
  \section{Bell Polynomial Identities}\label{bellpolyidsection}

 A number of Bell polynomial identities are needed in the sequel of this chapter. The identities we are most concerned with were introduced by D. Cvijovi\'c in \cite{cvijovic8} (2013). These identities are key ingredients in the arguments of \cite{karendiffeo}, they are also combinatorially significant \cite{bighunt}. In 2015, S. Eger re-proved some of these identities by translating into integer-valued distributions \cite{eger}. It is surprising, however, that elementary combinatorial proofs are actually quite applicable, and this section is devoted to displaying them. The reason that these proofs, despite being simple, were not discovered before is probably because the proofs are only seen clearly if the appropriate identity is chosen to start with.
 
\setlength{\parindent}{0cm}

\begin{lem}\label{lemma1}

Suppose $n,k>0$, then$$k\;B_{n,k}(x_1,x_2,\ldots)=\;\overset{n}{\underset{s=1}{\sum}}\;{n\choose s} \;x_s \;B_{n-s,k-1}(x_1,x_2,\ldots)\;,$$
and
\[n\;B_{n,k}(x_1,x_2,\ldots)=\;\overset{n}{\underset{s=1}{\sum}}\;{n\choose s}\;s\;x_s\;B_{n-s,k-1}(x_1,x_2,\ldots).\]
                                    \end{lem}
\begin{proof}For the first identity, the left hand side is the generating function for partitions with $k$ parts which are rooted at one part (localization). Seen another way, we may first choose $s$ arbitrary  elements from $\{1,2,\ldots,n\}$ to form our root part, and then generate all possible partitions with $k-1$ parts over the remaining $n-s$ elements, thus getting the right hand side.

For the second identity, the left hand side counts partitions with $k$ parts, which are rooted in a finer way than in the previous setting, namely, they are rooted at one of the $n$ elements. Again, we can do this rather differently (localizing in two levels): First choose $s$ special elements that will form the part which hosts the root, then choose the root from amongst them (in $s$ ways); finally generate all partitions with $k-1$ elements over the remaining elements, hence getting the right hand side.   \end{proof}

The next theorem is the main theorem in \cite{cvijovic8}. The proof presented here is new and does not make any reference to the analytic methods used in proving the identities in \cite{cvijovic8}. The proof only depends on the combinatorial meaning of Definition \ref{bellpoly}.

\begin{thm}\label{report}
The following Bell identities hold, where $B_{n,k}$ stands for $B_{n,k}(x_1,x_2,\ldots)$, the partial Bell polynomial with $k$ parts.
\begin{equation}
    B_{n,k}=\frac{1}{x_1}\cdot\frac{1}{n-k}\;\overset{n-k}{\underset{\alpha=1}{\sum}}\;{n\choose \alpha} \Big[(k+1)-\frac{n+1}{\alpha+1}\Big]\;x_{\alpha+1} \;B_{n-\alpha,k}\;\;,\label{id1}
\end{equation}

\begin{equation}
    B_{n,k_1+k_2}=\frac{k_1!\;k_2!}{(k_1+k_2)!}\;\overset{n}{\underset{\alpha=0}{\sum}}\;{n\choose \alpha} \;B_{\alpha,k_1}\;B_{n-\alpha,k_2}\;\;,\label{id2}
\end{equation}

\begin{equation}
    B_{n,k+1}=\frac{1}{(k+1)!}
\overset{n-1}{\underset{\alpha_1=k}{\sum}}
\overset{\alpha_1-1}{\underset{\alpha_2=k-1}{\sum}}\cdots
\overset{\alpha_{k-1}-1}{\underset{\alpha_k=1}{\sum}}
{n\choose \alpha_1}{\alpha_1\choose \alpha_2}\cdots{\alpha_{k-1}\choose \alpha_k} x_{n-\alpha_1}\cdots x_{\alpha_{k-1}-\alpha_k}x_{\alpha_k}.\label{id3}
\end{equation}
\end{thm}

\begin{proof}
Identity (\ref{id3}) is immediate from (\ref{id2}), so we start by proving (\ref{id1}). The following `starter' identity is clear from the definition of partial Bell polynomials:
\[B_{n+1,k+1}=\overset{n-k}{\underset{\alpha=0}{\sum}}\;{n\choose \alpha}\;x_{\alpha+1} \;B_{n-\alpha,k}\;.\]
Indeed, the identity exactly describes the natural passage from partitions of the set \;$\mathbb{N}_n=\{1,2,\ldots,n\}$ to partitions of $\mathbb{N}_{n+1}=\{1,2,\ldots,n,n+1\}$. Namely, to form all partitions of $\mathbb{N}_{n+1}$ with $k+1$ parts in which $n+1$ appears in a part of size $\alpha+1$, we can first choose $\alpha$ elements (in all possible ways) from $\mathbb{N}_n$ to be in the same part with $n+1$, and then generate all partitions with $k$ parts on the remaining $n-\alpha$ elements of $\mathbb{N}_n$. 

Now, by Lemma \ref{lemma1}, multiply both sides by $(k+1)$ to further get 
\[\overset{(n+1)-k}{\underset{s=1}{\sum}}\;{n+1\choose s}\;x_s\;B_{n+1-s,k}= (k+1)\;B_{n+1,k+1}=(k+1)\;\overset{n-k}{\underset{\alpha=0}{\sum}}\;{n\choose \alpha}\;x_{\alpha+1} \;B_{n-\alpha,k}\;.\]
Reindexing the left sum by $\alpha=s-1$, and explicitly writing the first term ($\alpha=0$) of both sides, we arrive at
\[(n+1) x_1B_{n,k}\;+\;\overset{n-k}{\underset{\alpha=1}{\sum}}\;{n+1\choose \alpha+1}x_{\alpha+1}B_{n-\alpha,k}= (k+1)x_1 B_{n,k}\;+\;(k+1)\overset{n-k}{\underset{\alpha=1}{\sum}}{n\choose \alpha}x_{\alpha+1} B_{n-\alpha,k}.\]

Hence,
\begin{align*}
(n-k) \;x_1\; B_{n,k}\;&=\;\overset{n-k}{\underset{\alpha=1}{\sum}}\;\Big[(k+1){n\choose \alpha}-{n+1\choose \alpha+1}\Big]\;x_{\alpha+1}\;B_{n-\alpha,k}\\&=\overset{n-k}{\underset{\alpha=1}{\sum}}\;{n\choose \alpha} \Big[(k+1)-\frac{n+1}{\alpha+1}\Big]\;x_{\alpha+1} \;B_{n-\alpha,k}\;,\end{align*}which gives identity (\ref{id1}).

Finally, identity \ref{id2} is actually easier. Given $k=k_1+k_2$, the generating function for partitions into $k$ parts with $k_1$ distinguished parts is given by ${k_1+k_2\choose k_1} B_{n,k_1+k_2}$. Another way is to first select $\alpha$ elements and use them to build a partition on $k_1$ parts (these are now naturally `highlighted' by this choice), and then generate a partition of the remaining $n-\alpha$ elements on $k_2$ parts.

\end{proof}

    \section{Generating Function Method}
  
The formula for $b_n$ in (\ref{form}) turns out to be exactly the compositional inverse of the diffeomorphism $F$. This can be seen through an old result that is mentioned in \cite[p.150-151]{comtet}, seemingly obtained independently by B\"{o}dewadt (1942) and Riordan (1968) among others. However, this was not recognized by the authors in \cite{karendiffeo}, and so highler level insights were obscured.
Now, let us take advantage of this fact. We derive a functional differential equation whose solution is the inverse of $F$. The idea is to rewrite equations (\ref{recurrence}) and (\ref{di}) so that we can apply the next Fa\`a di Bruno's composition of power series relation.

 \begin{lem}[Fa\`a di Bruno]\label{faa}
 Given two power series $f(t)=\overset{\infty}{\underset{n=0}{\sum}}\;f_n\displaystyle\frac{t^n}{n!}$ and $g(t)=\overset{\infty}{\underset{n=0}{\sum}}\;g_n\displaystyle\frac{t^n}{n!}$, the composition $\;h(t):=f(g(t))$ can be written as $h(t)=\overset{\infty}{\underset{n=0}{\sum}}h_n\displaystyle\frac{t^n}{n!},$ where
\[h_n=\overset{n}{\underset{k=0}{\sum}}\;f_k\; B_{n,k}(g_1,g_2,\ldots).\]
 \end{lem}

 The key is the next Proposition which gives differential equations for the exponential generating series of the $b_n$.
 \begin{prop}\label{1pm}Let $F(t)=\overset{\infty}{\underset{j=0}{\sum}}a_j\;t^{j+1}$ be a diffeomorphism of fields as before, and set $\;G(t):=\overset{\infty}{\underset{n=1}{\sum}}b_n\displaystyle\frac{t^n}{n!}$ where the $b_n$s satisfy the recurrences \ref{recurrence} and \ref{di}. Define 
 \begin{equation}\label{PandQ}
     Q(t):=\frac{1}{2} \;\frac{d}{dt}\Big(\big(F(t)\big)^2\Big)\;\;\;\;\;\text{and}\;\;\;\;\;P(t):=\int\Big(\frac{d}{dt}F(t)\Big)^2 dt\;.
 \end{equation}
 Then, on the level of generating functions, the recurrence (\ref{recurrence}) is equivalent to the differential equation 
 \begin{equation}
     0=t\;\frac{d}{dt}P\big(G(t)\big)\;-\; Q\big(G(t)\big)\;,  \label{myequation}
 \end{equation} and the recurrence (\ref{di}) is equivalent to the differential equation 
 \begin{equation}
     0=\;\frac{d^2}{dt^2}P\big(G(t)\big)\;+\; \frac{d^2G}{dt^2}\cdot\frac{d}{dG}P\big(G(t)\big)\;.\; \label{myequation2}
 \end{equation}\end{prop}

  \begin{proof}
 First we prove (\ref{myequation}): From the definitions, $P$ and $Q$ can be expanded as 
 \begin{equation}
     Q(t)=\frac{1}{2} \;\frac{d}{dt}\Big((F(t))^2\Big)=\overset{\infty}{\underset{k=1}{\sum}}\;\overbrace{\Big(k!\;\frac{(k+1)}{2}\;\;\overset{k-1}{\underset{j=0}{\sum}}a_ja_{k-1-j}\Big)}^{q_k}\;\frac{t^k}{k!}\;\;,
 \end{equation} and
  \begin{equation}
      P(t)=\int\Big(\frac{d}{dt}F(t)\Big)^2 dt=\overset{\infty}{\underset{k=1}{\sum}}\;\overbrace{\Big(k!\;\frac{1}{k}\;\overset{k-1}{\underset{j=0}{\sum}}a_ja_{k-1-j}\;(j+1)(k-j)\Big)}^{p_k}\;\frac{t^k}{k!}\;.
  \end{equation}
 Consequently, by the Fa\`a di Bruno's formula (Lemma \ref{faa}), we have 
 \begin{align}
     Q(G(t))&=\overset{\infty}{\underset{n=1}{\sum}}\;\Big(\overset{n}{\underset{k=1}{\sum}}\;q_k\;B_{n,k}(b_1,b_2,\ldots)\Big)\;\frac{t^n}{n!}\;,\\
     P(G(t))&=\overset{\infty}{\underset{n=1}{\sum}}\;\Big(\overset{n}{\underset{k=1}{\sum}}\;p_k\;B_{n,k}(b_1,b_2,\ldots)\Big)\;\frac{t^n}{n!}\;.
 \end{align}
In particular, 
  \[t\;\frac{d}{dt}P(G(t))=\overset{\infty}{\underset{n=1}{\sum}}\;n\;\Big(\overset{n}{\underset{k=1}{\sum}}\;p_k\;B_{n,k}(b_1,b_2,\ldots)\Big)\;\frac{t^n}{n!}\;.\]
 
 Then (\ref{myequation}) is given by
  \begin{align*}0&=t\;\frac{d}{dt}P\big(G(t)\big)\;-\; Q\big(G(t)\big)\\&=\overset{\infty}{\underset{n=1}{\sum}}\;\Big(\overset{n}{\underset{k=1}{\sum}}B_{n,k}(b_1,b_2,\ldots) \;(n\;p_k-q_k)\Big)\;\frac{t^n}{n!}.
 \end{align*}
 
 That is, equation (\ref{myequation}) is equivalent to the fact that for all $n\geq1$ 
  \begin{align*}0&=\overset{n}{\underset{k=1}{\sum}}B_{n,k}(b_1,b_2,\ldots) \;(n\;p_k-q_k)\\&=\overset{n}{\underset{k=1}{\sum}}B_{n,k}(b_1,b_2,\ldots)\left[n\Big(k!\frac{1}{k}\overset{k-1}{\underset{j=0}{\sum}}a_ja_{k-1-j}(j+1)(k-j)\Big)-\Big(k!\frac{(k+1)}{2}\overset{k-1}{\underset{j=0}{\sum}}a_ja_{k-1-j}\Big)\right]\\&=\overset{n}{\underset{k=1}{\sum}}B_{n,k}(b_1,b_2,\ldots)\frac{(k-1)!}{2}\;\;\overset{k-1}{\underset{j=0}{\sum}}a_ja_{k-1-j}\;\Big[\;2n(j+1)(k-j)-k(k+1)\;\Big],\end{align*}
 which is exactly the recurrence (\ref{recurrence}).\\

 Next we prove (\ref{myequation2}):
 
  We can cancel out the factor of $1/2$  and rearrange recurrence (\ref{di}) to obtain 
 \begin{align*}0=\;n(n-1)&\overset{n}{\underset{k=1}{\sum}}\;p_k\cdot\frac{1}{k}\overset{n-k+1}{\underset{s=1}{\sum}}\frac{b_s}{s!(n-s)!}B_{n-s,k-1}(b_1,b_2,\ldots)\;\\+\;&\overset{n}{\underset{k=1}{\sum}}\;p_k\;\overset{n-k+1}{\underset{s=1}{\sum}}\frac{b_s\;s(s-1)}{s!(n-s)!}B_{n-s,k-1}(b_1,b_2,\ldots)\;.\end{align*}
 
 Now, recall that
 \begin{align*}k\;B_{n,k}(b_1,b_2,\ldots)&=n!\;\overset{n}{\underset{s=0}{\sum}}\;\frac{b_s}{s!(n-s)!}B_{n-s,k-1}(b_1,b_2,\ldots)\\&=n!\overset{n-k+1}{\underset{s=1}{\sum}}\frac{b_s}{s!(n-s)!}B_{n-s,k-1}(b_1,b_2,\ldots)\;,\end{align*}
 since $b_0=0$ by convention, and since Bell polynomials vanish whenever the number of parts is greater than the size. The latter reason is exactly what allows us to also change the bounds of the summations in the above equation to finally get  
  \begin{align*}0&=\;\frac{n(n-1)}{n!}\overset{n}{\underset{k=1}{\sum}}\;p_k\;B_{n,k}(b_1,b_2,\ldots)\;\\&+\;\overset{n}{\underset{s=0}{\sum}}\;\frac{s(s-1)\;b_s  }{s!(n-s)!}\overset{n-s+1}{\underset{k=1}{\sum}}p_k \;B_{n-s,k-1}(b_1,b_2,\ldots)\\&=[t^n]\; t^2\;\frac{d^2}{dt^2}P\big(G(t)\big) \;+\; \overset{n}{\underset{s=0}{\sum}} [t^s] \big(t^2\frac{d^2G}{dt^2}\big) \cdot    [t^{n-s}] \Bigg(\frac{dP}{dt}\Bigg)(G),\end{align*}
 which establishes the claim (notice that $ \frac{dP}{dt}(G(t))= \frac{d}{dG}P(G)$ ).\\
 \end{proof}

 \begin{cor}\label{corcorcor}The compositional inverse $F^{-1}$ of the diffeomorphism $F$ is a solution for (\ref{myequation}) and (\ref{myequation2}).\end{cor}
 
 \begin{proof}

 Equation (\ref{myequation}) can be simplified further: 
 \begin{align*}
0&=t\;\frac{d}{dt}P\big(G(t)\big)\;-\; Q\big(G(t)\big)\\&=t\;\frac{d}{dG}P\big(G(t)\big)\;\frac{dG}{dt}-\; \frac{1}{2}\frac{d}{dG}\Big(F\big(G(t)\big)\Big)^2\\&=t\;\Big(\frac{d}{dG}F\big(G(t)\big)\Big)^2\;\frac{dG}{dt}\;-\;F\big(G(t)\big)\;\frac{d}{dG}F\big(G(t)\big).
\end{align*} 
 
 Now we might assume that $\frac{d}{dG}F\big(G(t)\big)\neq 0$, that is just $F(G(t))$ is not a constant, and then we have 

\[0=t\;\frac{d}{dG}F\big(G(t)\big)\;\frac{dG}{dt}\;-\;F\big(G(t)\big).\]

 That is to say,
 \[F\big(G(t)\big)=t\;\frac{d}{dt}F\big(G(t)\big),\] and a simple separation of variables then gives that $G=F^{-1}$ is a solution.
 
 \bigskip
 \bigskip
 Proceding to Equation (\ref{myequation2}) which boils down maybe more insightfully:
 
 \begin{align*}
 0&=\;\frac{d^2}{dt^2}P\big(G(t)\big)\;+\; \frac{d^2G}{dt^2}\cdot\frac{d}{dG}P\big(G(t)\big)\\
  &=\;\frac{d}{dt}\Big(\frac{d}{dG}P\big(G(t)\big)\cdot\frac{dG}{dt}\Big)\;+\; \frac{d^2G}{dt^2}\cdot\frac{d}{dG}P\big(G(t)\big)\\
  &=2\;\frac{d^2G}{dt^2}\cdot\frac{d}{dG}P\big(G(t)\big)\;+\;\frac{d}{dt}\Big(\frac{d}{dG}P\big(G(t)\big)\Big)\cdot\frac{dG}{dt}\\
  &=2\;\frac{d^2G}{dt^2}\cdot\Big(\frac{d}{dG}F\big(G(t)\big)\Big)^2\;+\;\frac{d}{dt}\Bigg(\Big(\frac{d}{dG}F\big(G(t)\big)\Big)^2\Bigg)\cdot\frac{dG}{dt}\\
  &=2\;\frac{d^2G}{dt^2}\cdot\Big(\frac{d}{dG}F\big(G(t)\big)\Big)^2\;+\;2\;\frac{d}{dG}F\big(G(t)\big)\cdot\frac{d}{dG}\Big(\frac{d}{dG}F\big(G(t)\big)\Big)\cdot\Big(\frac{dG}{dt}\Big)^2.
 \end{align*}
 
Again we assume that $G(t)$ is not a constant, and so neither is $F(G(t))$. Hence,

\begin{align*}
 0&=\frac{d^2G}{dt^2}\cdot\frac{d}{dG}F\big(G(t)\big)\;+\;\frac{d^2}{dG^2}F\big(G(t)\big)\cdot\Big(\frac{dG}{dt}\Big)^2\\
  &=\frac{d}{dG}\Big(\frac{dG}{dt}\Big)\cdot\frac{dG}{dt}\cdot\frac{d}{dG}F\big(G(t)\big)\;+\;\frac{d^2}{dG^2}F\big(G(t)\big)\cdot\Big(\frac{dG}{dt}\Big)^2,
 \end{align*}
  
  Then by our assumption above
 \begin{align*}
 0&=\frac{d}{dG}\Big(\frac{dG}{dt}\Big)\cdot\frac{d}{dG}F\big(G(t)\big)\;+\;\frac{d^2}{dG^2}F\big(G(t)\big)\cdot\frac{dG}{dt}\\&=\frac{d}{dG}\Big[\frac{dG}{dt}\cdot\frac{d}{dG}F\big(G(t)\big)\Big]\\&=\frac{d}{dG}\Bigg(\frac{d}{dt}F\big(G(t)\big)\Bigg).
 \end{align*}
 
 Again this leads to that $G=F^{-1}$ is a solution, which proves the corollary.
 \end{proof}

 Thus, we have shown that the series whose coefficients are the solution $(b_n)$ of  recurrences (\ref{recurrence}) and (\ref{di})  is exactly the compositional inverse of the diffeomorphism $F$. This is highly suggestive of a connection to the Legendre transform, since, as we shall see in Section~\ref{sec legendre}, the definition of the combinatorial Legendre transform goes through an almost identical process of summing over tree graphs and involves compositional inverses \cite{kjm1,kjm2}.

\section{Generalization to scalar interacting fields}

In this section we overview the application of such diffeomorphisms as before to theories with interaction terms. In particular, for the sake of completeness, we are interested in displaying some of the work in \cite{Paulthesis} which has been significantly simplified in \cite{paulrefers} by means of our generating functions method. For the sake of clarity and unification, we will try to use the same notation as in \cite{paulrefers} as much as possible. Assume we are given multiple interaction terms in the Lagrangian:
\[L(\phi)=\frac{1}{2}\partial_\mu\phi(x)\partial^\mu\phi(x)-\frac{m^2}{2}\phi^2(x)-
\sum_{s=3}^\infty\frac{\lambda_s}{s!}\phi^s(x).\]
Each interacting term results in an additional type of vertex when the diffeomorphism $F$ is applied, with corresponding Feynman rules $w_n^{(s)}=0 $ for $n<s$ and 
\begin{equation}
    -iw_n^{(s)}=-i\frac{\lambda_s}{s!}n!\underset{j_1+\cdots+j_{s}=n}{\sum_{j_1,\cdots,j_{s}}}a_{j_1}\cdots a_{j_{s}}=-i\lambda_s B_{n,s}(1!a_0,2!a_1,3!a_2,\ldots), \label{newinteraction}
\end{equation} for all $n\geq s$. In particular, the interaction terms add
$-iw_n:=-i\sum_s w_n^{(s)}$ to the $n$-valent interaction vertex. Let $W_n^{(s)}$ be the sum of contributions of trees with $n$ external edges and only one vertex of type $-iw_j^{(s)}$. Then it is calculated in \cite{paulrefers} that 

\begin{equation}
    W_n^{(s)}=-i\lambda_s\sum_{k=s}^n B_{k,s}(1!,2!a_1,3!a_2,\ldots)B_{n,k}(b_1,b_2,\ldots). \label{Smatrix}
\end{equation}

 \begin{thm}[Theorem 3.3 in \cite{Paul}]
In terms of the $S$-matrix elements, $W^{(s)}_s=-i\lambda_s$ and $W_n^{(s)}=0$  for $n\neq s$.
 \end{thm}
 
 \begin{proof}
The proof completely follows by calling the Fa\`a di Bruno relations, and recalling that $F^{-1}(\phi)=\sum_n b_n \frac{\phi^n}{n!}$ which is the same as $\phi(F)=F^{-1}(F(\phi))=\sum_n b_n\frac{F(\phi)}{n!}$. Namely,
\[W_n^{(s)}=-i\lambda_s\sum_{k=s}^n k! [\phi^k] (F^s(\phi)) B_{n,k}(b_1,b_2,\ldots)=-i\lambda_s [F^n] F(\phi(F))^s=-i\lambda_s [F^n]F^s,\] which establishes the statement.
 \end{proof}

\section{Relation with the Legendre Transform}\label{sec legendre}

Given a function $f:x\mapsto f(x)$, it might be desirable, in many contexts, to express every thing in terms of $y=f'(x)$ instead of $x$, without losing information about the function. The (analytic) Legendre transform does indeed achieve this goal but for a restricted class of functions, namely convex smooth functions. However, physicists usually use the Legendre transform even when the functions involved fail to satisfy these requirements. The surprise is when such calculations match with experimental results. In \cite{kjm1}, a combinatorial Legendre transform is defined which generalizes the analytic version and unveils the hidden robust algebraic structure of the Legendre transform. Furthermore, as has been long known to physicists, the Legendre transform builds the series of trees, so as well as making sense rigorously as a formal power series operator, the Legendre transform has a combinatorial meaning as an operator on generating series.

The defining relation for the Legendre transform $\mathrm{L}$ is given as \[\mathrm{L}f(z)=-z g(z)+ (f\circ g)(z),\]

where $g$ is the compositional inverse of the derivative, namely, $x=g(y)$.

The idea in \cite{kjm1} was to realize this relation as coming from the following combinatorial bijection between classes:

\[v_1 \partial_{v_1} \mathcal{T}^l\uplus\biguplus_{k\geq 2} v_k\circ (\partial_{v_1}\mathcal{T}^l)\sim \mathcal{T}^l\uplus (e\ast(\mathcal{E}_2\circ e^{-1}\partial_{v_1} \mathcal{T}^l),
\]
where $\mathcal{T}^l$ is the class of labelled trees, and $v_1\partial_{v_1}$ is the operation of rooting or distinguishing at a 1-vertex (vertex of degree 1), $e,e^{-1}$ for edges and anti-edges (the inverse with respect to the gluing operation) respectively. 

\begin{dfn}[Combinatorial Legendre Transform \cite{kjm2}]If $A\in R[[x]]$ and $A'^{-1}$ exists, then the \textit{combinatorial Legendre transform} is defined to be \[(\mathrm{L}A)(x)=(A\circ A'^{-1})(x)-x A'^{-1}(x).\]\end{dfn}
\begin{thm}[\cite{kjm2}]

Let $T_A(y)= [[\mathcal{T}^l,\omega_{v_1}\otimes\omega_e\otimes\omega]](y,u,\lambda_2,\lambda_3,\ldots)$, be the generating series of labelled tree graphs with indeterminates $y,u,\lambda_i$ standing for leaves, edges and vertices of higher orders, respectively; and where  $A(x)=-u^{-1}\frac{x^2}{2!}+\sum_{k\geq3} \lambda_k \frac{x^k}{k!}$. Then $(\mathrm{L}A)(y)=T_A(-y)$. That is, the Legendre transform is exactly the generating series of tree graphs with the prescribed weights and variables.
\end{thm}

This combinatorial Legendre transform is applicable to all formal power series with vanishing constant and linear terms, and the most important fact is that the new Legendre transform drops the convexity constraint. The following formula for compositional inverses appears in \cite{loday} (compare to the formula in \ref{form}).

\begin{lem}[\cite{loday}] If $a(x)=x+a_1x^2+a_2x^3\cdots$ then the compositional inverse of $a(x)$ is $r(x)=x+r_1x^2+r_2x^3\cdots$ where 
\[r_n=\sum_{m_1,\ldots,m_k}(-1)^{\sum m_j} B(n+1,m_1,\ldots,m_k)\;a_1^{m_1}\cdots a_k^{m_k},\] where $B(m_{-1},m_1,\ldots,m_k)$ is the number of unlabelled rooted trees with $m_j$ vertices with $j+1$ children.
\end{lem}
Let $A(x)=-\frac{x^2}{2!}+\sum_{k\geq3} \lambda_k \frac{x^k}{k!}$  and $y=-\partial A/\partial x$. Then it can be shown that $\mathrm{L}A(y)=y^2/2+\sum_{n=3}L_n \frac{y^n}{n!}$ with $L_n=(n-1)!r_{n-2}$, where $r(x)$ is the inverse of $y(x)$ (see \cite{kjm1}). Now, given a diffeomorphism $F(\phi)=\overset{\infty}{\underset{j=0}\sum} a_j\phi^{j+1}$ (with $a_0=1$), then $F=-\displaystyle\frac{\partial A}{\partial \phi}$, where 
\begin{equation}\label{eq A}
A=-(\displaystyle\frac{\phi^2}{2}+\displaystyle\frac{a_1}{3}\phi^3+\displaystyle\frac{a_2}{4}\phi^4+\cdots).
\end{equation}
With our previous information, the inverse of $F$ is $\sum_{n=1} \frac{b_n}{n!}$, thus, the $n$th coefficient in $\mathrm{L}A(y)$ is 
\[\displaystyle\frac{L^{(n)}}{n!}=\displaystyle\frac{(n-1)!}{n!}\cdot\displaystyle\frac{b_{n-1}}{(n-1)!}=\frac{b_{n-1}}{n!}.\]

Let us now compare the Legendre transform context with the main situation of the present paper.  For us, the series with coefficients $b_{n-1}/n!$ is the series for tree diagrams with only one external edge off-shell in the theory transformed by the diffeomorphism $F$.   The Lagrangian for this theory $L_F(\phi)$ is as given at the beginning of Section~\ref{sec field theory set up}.  This is a much more complicated Lagrangian than the $A$ of \eqref{eq A}.  However, the Legendre transform of $A$ gives the series with coefficients $b_{n-1}/n!$, as we saw above using only the combinatorial Legendre transform results.  But, the Legendre transform of $L_F$ must give the series of tree amplitudes of the theory with Lagrangian $L_F$, and this is essentially, by the key results we have been discussing here, also the series with coefficients $b_{n-1}/n!$.  

So we have two apparently very different theories giving the same tree series.  This feels somewhat like the initial set up where the free theory and the transformed theory looked very different, but ultimately did, as one would naively expect, describe the same theory.  The situation here, however, is a bit different: $A$ really is a much simpler Lagrangian which cannot describe the same theory as it is purely a graph-counting-type Lagrangian, with no physical parameters at all.  The solution to this discrepancy is hidden in the `essentially' at the end of the previous paragraph.  Namely, the series with coefficients $b_{n-1}/n!$ is not the series of tree amplitudes of the theory with Lagrangian $L_F$, but the series of the tree amplitudes \emph{with all but one external edge put on-shell}.  The combinatorial Legendre transform can not see the on-shell condition -- momentum conservation at the vertices is built in, but nothing further regarding the physical nature of the momenta.  

Thus we see that the simplifications brought by putting all by one external edge on-shell bring us to a tree level theory that can also be obtained as the tree level outcome of the Lagrangian $A$ without any kinematical conditions.  By comparison, putting all of the external edges on-shell simplifies even further, giving $0$ for all higher point functions, that is, recovering the free theory, the proof of which was the original goal of \cite{karendiffeo}.

\section{Discussion}

We will conclude with a few comments.

First note that as discussed in \cite{karendiffeo} the loop level results follow from these tree level results by reinterpreting the $b_n$ as describing trees with at most one off-shell external edge and then gluing.  For this to work the renormalization scheme must be a kinematical scheme, such as a subtraction scheme.  The key property needed of the renormalization scheme is that all tadpoles vanish.
We continue to find this interesting evidence that the kinematical schemes are physically better in some sense, and we continue to wonder if the loop-level results have a more fully combinatorial derivation as well.

There remains some subtleties in how exactly the renormalization process interacts with the diffeomorphism.  At the time of this work, we did not have a theorem to say that the framework given here, along the with the vanishing of tadpoles should suffice to make it possible for the renormalization process interact well with the diffeomorphism.  As mathematicians, we were reluctant to speculate much beyond the theorems we could prove.  Subsequently, however, this has been investigated by Paul Balduf, and now there are some results, see in particular Theorem 4.1 of \cite{Paulrenorm}.

The generating functions proof highlights the importance of the equations \eqref{myequation} and \eqref{myequation2} as they are the differential equations analogous to the two defining recurrences for the $b_n$s.  Many other differential equations would also have $G=F^{-1}$ as a solution; is there some structural feature that distinguishes \eqref{myequation} and \eqref{myequation2}, or conversely could other equations with this solution have similar physical meanings when converted back into recurrences?

One might also approach the proof using an infinitesimal transformation $\phi \mapsto \phi + \epsilon F(\phi) + O(\epsilon^2)$ and work modulo $\epsilon^2$.\footnote{Thanks to the referee for this suggestion}  For instance this is the approach of \cite{johnsonfreyd2010coordinate}.  In our situation this would mean that massive and kinematic vertices of all degrees would still be generated, but the Feynman rules for them would be much simpler $i(n-1)!a_{n-2}(p_1^2+p_2^2+\cdots+p_n^2)$ for the kinematic vertex of degree $n$ and $-im^2n!a_{n-2}$ for the massive vertex of degree $n$.  Likewise in the Bell polynomial expression for $b_{n+1}$, each partial Bell polynomial would only contribute one term modulo $\epsilon^2$ giving $b_{n+1} = \sum_{k=1}^n\frac{(n+k)!}{(k-1)!}(-1)^ka_{n-k+1}\epsilon$ modulo $\epsilon^2$ and \eqref{eq orig b rec} would also have only the extreme terms in the inner sum, much simplifying the proof that \eqref{eq orig b rec} implies the previous expression.  

However, for our approach of working purely diagrammatically with no path integral assumptions whatsoever, we cannot take this infinitesimal simplification.  The freedom to work modulo $\epsilon^2$ is fundamentally a path integral assumption -- a very mild one that surely any reasonable path integral would satisfy, but none the less a path integral assumption.  By forgoing the path integral entirely, we also forgo this simplification.  The resulting extra complexity is not without some benefits, we see more clearly how the functional inverse comes in and we get to prove some nice Bell polynomial identities along the way as well as maintaining a complete independence from any path integral assumptions.

\medskip

Two obvious directions for future inquiry are field diffeomorphisms of theories beyond scalar theories as well as more general transformations going beyond these point transformations of the field.  For the former we do not expect much difficulty, though the details are bound to be interesting.  The latter seems more difficult, at least it is not evident to us what other good combinatorial questions can be asked in that direction, so it becomes a very different question.

From the present results one knows that all scalar theories related by a field diffeomorphism are equivalent, so one has equivalence classes of theories.  Any of the generalizations above would also lead to equivalence classes of theories.  Some of the initial motivation of \cite{kreimer2} was the idea that perhaps some of the more vexing field theories might have simpler theories in the same equivalence class, while in \cite{karendiffeo} there are some thoughts regarding equivalence classes and Haag's theorem.  Ultimately, equivalence classes under only field diffeomorphism are too small for most purposes.  For instance one doesn't get the whole Borchers class of even the free field \cite{EpsteinH1963Otbc}.  However, while the part of the Borchers class obtained by field diffeomorphisms is small, it is up to scaling the invertible part, in the sense that differomorphisms are invertible and any linear combination of powers of $\phi$ which does not involve $\phi$ itself cannot be transformed by a formal series to re-obtain $\phi$ and $\phi$ cannot be re-obtained from its derivative using only further derivatives and powers.  One could investigate an analogous combinatorial analysis of field transformation in the context of Ore algebras rather than algebras of formal power series in order to to incorporate the derivatives.

The connection with the Legendre transform also brings up the question of when some restrictions in one theory, such having all but one external edge on-shell, lead to something equivalent to a simpler theory without such restrictions, and what is the physical meaning of this.  One way to understand the physical meaning of why the tree series of $L_F$ with all but one external edge on-shell corresponds to the tree series of $A$ is given by Paul Balduf in his physical proof of the appearance of the inverse of $F$, see Lemma 2.22 of \cite{Paul}, which is done by moving to position space and using the vanishing of tadpoles.

\bibliographystyle{plain}
\bibliography{references.bib}

\end{document}